\documentclass{amsart}
\usepackage{amsaddr}
\usepackage{amssymb}
\usepackage{amsmath,amsthm,nicematrix}
\usepackage{amsfonts}
\usepackage{graphicx}
\usepackage[top = 1in, right = 1in, left=1in, bottom = 1in]{geometry}
\usepackage{setspace}
\usepackage{xcolor}

\usepackage[
backend = biber,
sorting = nyt,
style = numeric
]{biblatex}
\usepackage{hyperref}

\newcommand{\C}{\mathcal{C}}
\newcommand{\F}{\mathbb{F}}
\newcommand{\bs}{\boldsymbol}

\DeclareMathOperator{\nRS}{nRS}

\DeclareMathOperator{\gcrd}{gcrd}

\let\mod\relax
\DeclareMathOperator{\mod}{\ mod}

\newtheorem{theorem}{Theorem}[section]
\newtheorem{proposition}[theorem]{Proposition}
\newtheorem{corollary}[theorem]{Corollary}
\newtheorem{lemma}[theorem]{Lemma}

\theoremstyle{definition}
\newtheorem{example}[theorem]{Example}
\newtheorem{definition}[theorem]{Definition}

\theoremstyle{remark}
\newtheorem{remark}{Remark}

\begin{document}
\title{Nonlinear Skew Quasi-Cyclic Codes}
\author{Daniel Bossaller}
\address[Bossaller]{University of Alabama in Huntsville, Huntsville, AL, USA}
\email{daniel.bossaller@uah.edu}
\author{Daniel Herden}
\address[Herden]{Baylor University, Waco, TX, USA}
\email{daniel\_herden@baylor.edu}
\author{Indalecio Ruiz-Bola\~nos}
\email{inda.ruiz@txstate.edu}
\address[Ruiz-Bola\~nos]{Texas State University, San Marcos, TX, USA}

\thanks{The second author was supported by Simons Foundation grant MPS-TSM-00007788.}

\begin{abstract}
    This  article explores nonlinear analogues of skew quasi-cyclic codes of index~$\ell$, i.e., $\F_{q^m}[X;\sigma]$-submodules of $\left(\F_{q^m}[X;\sigma]/(X^n - 1)\right)^\ell$. After introducing nonlinear skew quasi-cyclic codes, we then determine the module structure of these codes by using a two-fold iteration of the Smith normal form of matrices over skew polynomial rings. We show that actually a single use of the Smith normal form will suffice to determine the elementary divisors of the code. Along the way, we also describe duals of our codes with respect to appropriately chosen inner products.
\end{abstract}

\maketitle
\section{Introduction}\label{intro}
This article investigates nonlinear generalizations of Reed-Solomon codes and skew quasi-cyclic codes. The study of error-correcting codes is a centerpiece of modern digital communication and storage systems, ensuring reliable transmission and recovery of data. Among the most significant developments in this field are Reed-Solomon (RS) codes, renowned for their robust error-correcting capabilities. RS codes were introduced by Irving S. Reed and Gustave Solomon in 1960 \cite{rs1960} to address the need for robust error correction in digital communication and storage. Initially hampered by computational complexity, their adoption surged with the development of efficient decoding algorithms like the Berlekamp-Massey algorithm.

This paper introduces and explores nonlinear skew quasi-cyclic codes, a generalization of skew cyclic \cite{Boucher-Geiselmann-Ulmer:2007, Boucher-Ulmer:2009, Gluesing-Luerssen:2021} and skew quasi-cyclic codes \cite{Abualrub-et-al:2010, Abualrub-et-al:2018, Bhaintwal:2012, Gao-Shen-Fu:2016} along the lines of $\F_q$-linear $\F_{q^t}$-codes \cite{Huffman:2010, Huffman:2013, Sharma-Kaur:2017}.

A linear code $\C$ of length $n\ell$ over $\F_{q^m}$ is called a \emph{quasi-cyclic code of index~$\ell$} if  it is invariant under the shift of codewords by $\ell$ positions. These codes can be studied as $\F_{q^m}[X] / (X^n-1)$-modules by considering their polynomial representation, where $\F_{q^m}[X]$ is the ring of polynomials with coefficients in $\F_{q^m}$ with the usual commutative addition and multiplication.

Given an $\F_q$-automorphism $\sigma: \F_{q^m}\rightarrow \F_{q^m}$, the ring of skew polynomials $\F_{q^m}[X;\sigma]$ is the collection of polynomials with the usual addition but multiplication defined by $X\cdot a = \sigma(a)\cdot X$ for all $a\in \F_{q^m}$. The multiplication in $\F_{q^m}[X;\sigma]$ is not commutative. A linear code $\C$ of length $n$ over $\F_{q^m}$ is called \emph{skew cyclic} if $(c_0, c_1, \ldots, c_{n-1})\in \C$ implies that $\left(\sigma(c_{n-1}),\sigma(c_0), \ldots, \sigma(c_{n-2})\right)\in \C$. Skew cyclic codes can be studied as left ideals of the ring $\F_{q^m}[X;\sigma] / (X^n-1)$ when their codewords are regarded as skew polynomials, provided that $m\mid n$.

Our article is organized as follows. Section \ref{sec:prelim} introduces the mathematical structures that are relevant to the present article. Section \ref{sec:nrs} examines \emph{additive Reed-Solomon codes}, as introduced by \cite{Yadav-2024}. This investigation leads naturally to the general concept of \emph{nonlinear skew quasi-cyclic codes}, which is the main subject of this article. Section \ref{sec:polycharacterization} introduces the quotient skew polynomial rings $R_n$ and $P_n$, which are then used to describe $\F_{q^a}$-linear skew quasi-cyclic $\F_{q^m}$-codes as $P_n$-modules.

Finding the algebraic structure of codes is an important problem \cite{Lally-Fitzpatrick:2001, Ling-Sole:2006, Ling-Sole:2001, Ling-Sole:2003, Ling-Sole:2005, Ou-azzou-et-al:2024}. In Section~\ref{sec:structure1} we analyze the module structure of $\F_{q^a}[X;\sigma]$-modules using the Smith normal form at two different instances. Consider two $\F_{q^a}[X;\sigma]$-submodules $C$ and $D$ such that $D\leq C\leq \F_{q^m}[X;\sigma]^\ell$, and obtain a basis $\boldsymbol{c}_1^*, \boldsymbol{c}_2^*, \ldots, \boldsymbol{c}_{\xi}^*$ for $C$ such that $ \left\{  d_i \boldsymbol{c}_i^* : i =1, \ldots, \zeta  \right\} $ with $\zeta \le \xi$ is a basis of~$D$. Then we provide the $\F_{q^a}[X;\sigma]$-module structure of $C\big / D$.

The process of Section~\ref{sec:structure1} utilized two applications of the Smith normal form in order to find the structure constants of $\F_{q^a}$-linear skew quasi-cyclic $\F_{q^m}$-codes $\C$. However, the theory used was much more general than is necessary for the current setting of nonlinear skew quasi-cyclic codes. Section~\ref{sec:totaldivisors} is devoted to the classification of \emph{total divisors} of skew polynomials, culminating in Theorem~\ref{thm:grcdmakesmagic}, which allows us to, in Section~\ref{sec:structure2}, use only one application of the Smith normal form to find a stacked basis $\boldsymbol{q}_1, \boldsymbol{q}_2,\ldots, \boldsymbol{q}_{r\ell}$ of $\F_{q^a}[X;\sigma]^{r\ell}$ and total divisors of $X^n-1$. These are used to obtain the $P_n$-module structure of $\C$ and to find an alternative generating set for $\C$. Additionally, this stacked basis allows us to define a new inner product and dual space for $\C$.

\section{Preliminaries}\label{sec:prelim}
As the present article presents generalizations of cyclic and skew cyclic codes, we will use this section to review some of the relevant background. We will, for the most part, state results without proof; the interested reader may refer to standard textbooks such as \cite{Macwilliams-Sloane:1977} or to the cited articles.

\subsection{Cyclic and Quasi-Cyclic Codes}
In this subsection let us fix a finite field $\F_q$ and some positive integer $n$ such that $\gcd(n,q) = 1$. Then we will fix the following notation $R_n := \F_q[X] / (X^n - 1)$ for the ring of degree $<n$ polynomials with the relation $X^n = 1$.

\begin{definition}
    A code $\C$ over $\F_q$ of length $n$ is called {\it cyclic} if whenever $(c_0, c_1, \ldots, c_{n-1})$ is contained in $\C$, then $(c_{n-1}, c_1, \ldots, c_{n-2}) \in \C$, that is $\C$ is invariant under cyclic shifts. It is straightforward to show that cyclic codes are in one-to-one correspondence with the ideals of $R_n$.
\end{definition}

As $R_n$ is a commutative principal ideal ring, each cyclic code $\C$ contains a unique monic polynomial of smallest degree $g(X)$ (called the \emph{generator polynomial} of $\C$) such that the ideal generated by $g(X)$ is $\C$. It is straightforward (using the division algorithm) to see that $g(X) \mid (X^n - 1)$. Hence, define the polynomial $h(X) : = \frac{X^n -1}{g(X)}$. Note that $h$ is a polynomial of degree $n - \deg(g)$.

\begin{lemma}
    Let $\C$ be a cyclic code with generator polynomial $g(X)$ of degree $n-k$, and suppose that $h(X)$ is defined as above, then:
    \begin{enumerate}
        \item[$(1)$] $\C$ is a linear code of dimension $k$.
        \item[$(2)$] Let $\hat h(X) : = X^{k}h(X^{-1})$ be the \emph{reciprocal polynomial} of $h(X)$. Then $\hat h(X)$ is a generator polynomial for $\C^\perp$ under the standard Euclidean inner product.
    \end{enumerate}
\end{lemma}

\begin{definition}
    \emph{Quasi-cyclic codes} are generalizations of cyclic codes which only repeat after cyclically shifting codewords some $\ell \geq 1$ positions. In other words, for a word of length $n = m \cdot \ell$, if
    \[(c_0, c_1, \ldots, c_{\ell-1}, c_\ell, \ldots, c_{m \ell-1}) \in \C,\] then
    \[(c_{(m-1) \ell}, \ldots, c_{m \ell-1}, c_0, \ldots, c_{(m-1) \ell - 1}) \in \C.\]
\end{definition}
The reader is encouraged to consult \cite{Ling-Sole:2006, Ling-Sole:2001,Ling-Sole:2003, Ling-Sole:2005} for a wide-ranging overview of quasi-cyclic codes.

The invariance under $\ell$-shifts of the codewords permits a simpler notation of a word in a quasi-cyclic code as a matrix, as follows:
\[\begin{pmatrix}
        c_0           & c_1               & \cdots & c_{\ell-1}     \\
        c_{\ell}      & c_{\ell+1}        & \cdots & c_{2\ell - 1}  \\
        \vdots        & \vdots            & \ddots & \vdots         \\
        c_{(m-1)\ell} & c_{(m-1)\ell + 1} & \cdots & c_{m \ell - 1}
    \end{pmatrix}\]
When the codewords are written in this way, it is evident that the code is invariant under cyclic ``row-shifts.'' This notation also suggests that every quasi-cyclic code is composed of cyclic codes, as is borne out in the following result.

\begin{lemma}
    Let $R_m$ denote $\F_{q}[X]/(X^m -1)$. Then the quasi-cyclic codes of index $\ell$ and length $m \cdot \ell$ are in one-to-one correspondence with the $R_m$-submodules of the free $R_m$-module of rank $\ell$, $R_m^\ell$.
\end{lemma}

\subsection{Skew Polynomials and Skew (Quasi-)Cyclic Codes}
Skew polynomials are generalizations of polynomials first explored by Ore in \cite{Ore:1933}; in that article, the author showed that skew polynomials (to be formally defined below) are the most general polynomial-type ring for which the following two familiar properties of polynomials still hold:
\begin{enumerate}
    \item[$(1)$] $\deg(f(X) + g(X)) \leq \max\{\deg(f(X)), \deg(g(X))\}$,
    \item[$(2)$] $\deg(f(X)g(X)) = \deg(f(X)) + \deg(g(X))$.
\end{enumerate}
These two properties together assure that there must exist some type of division algorithm and Euclidean algorithm.

\begin{definition}
    Let $R$ be a (possibly noncommutative) ring and $\sigma: R \to R$ an endomorphism of $R$. Then the \emph{skew polynomial ring} $R[X;\sigma]$ is the (left) $R$-algebra of formal polynomials
    \[f(X) = \sum_{i = 0}^n a_i X^i\] with $a_i \in R$ with the addition defined component-wise and multiplication defined by the equation
    \[Xa = \sigma(a)X.\]
    In the present article we assume that $R = \F_{q^m}$ is a finite field for $q$ a prime power, and that $\sigma$ is an $\F_{q}$-linear field automorphism.
\end{definition}

Skew polynomial rings have been extensively studied by many authors, particularly, Lam and Leroy in \cite{Lam-Leroy:1988}. The study of coding-theoretic applications of skew polynomials was initiated in 2009 by Boucher and Ulmer, \cite{Boucher-Ulmer:2009}, who examined skew polynomials as the basis for skew cyclic codes. Skew polynomials retain many familiar properties of traditional polynomials in a noncommutative setting. In the next proposition we list some of these properties. The reader is encouraged to consult \cite{Ore:1933} for definitions and proofs of the following facts.
\begin{proposition}
    Let $R$ denote the skew polynomial ring $\F_{q^m}[X;\sigma]$, and say that $f(X)$ and $g(X)$ are skew polynomials with $g(X)$ nonzero.
    Then the following holds.
    \begin{enumerate}
        \item[$(1)$] $R$ is a right Euclidean domain.

        \item[$(2)$] Given any two skew polynomials $f(X)$ and $g(X)$ in $R$, there is a unique monic \emph{greatest common right divisor} $\gcrd(f(X), g(X)) := d(X) \in R$. Moreover, there are polynomials $a(X)$ and $b(X)$ so that
              \[d(X) = a(X) f(X) + b(X) g(X).\]
    \end{enumerate}
\end{proposition}

\begin{remark}
    In the above proposition, similar statements hold if we switch `left' for `right,' and vice-versa. For the purposes of this article, we will choose to work with $R$ as a right Euclidean domain as in the above proposition; however, the arguments of this paper also work mutatis mutandis for left Euclidean domains.
\end{remark}

Finally we note the following property of skew polynomials which will be useful in the exploration of skew cyclic and skew quasi-cyclic codes.

\begin{lemma}
    Suppose that $n > 0$ with $m \mid n$, and that $\sigma$ is an $\F_q$-auto\-mor\-phism of $\F_{q^m}$. Then the polynomial ${X^n - 1}$ commutes with every skew polynomial. As a result $R_n := \F_{q^m}[X; \sigma]/(X^n - 1)$ is a ring.
\end{lemma}

\begin{definition}
    Let $\F_{q^m}$ be a finite field with $q$ a prime power, let $\sigma$ an $\F_q$-automorphism of $\F_{q^m}$ and $m\mid n$. A \emph{$\sigma$-skew cyclic code} $\C$ of length $n$ is then a left ideal of the ring $R_n = \F_{q^m}[X, \sigma] / (X^n - 1)$. We will usually write \emph{skew cyclic code} (instead of $\sigma$-skew cyclic code) whenever $\sigma$ is clear from the context.
\end{definition}

\begin{remark}
    Because skew polynomial rings admit a (one-sided) division algorithm, we may define generalizations of skew cyclic codes (and skew quasi-cyclic codes) also for the case $m\nmid n$; however, we must navigate the submodules rather than the ideals of the left $\F_{q^m}[X;\sigma]$-module \[R_n = \F_{q^m}[X;\sigma]/(X^n - 1),\] where $(X^n - 1)$ denotes the left $\F_{q^m}[X; \sigma]$-module generated by the skew polynomial $X^n - 1$.
\end{remark}

Using the relation $X^n = 1$ that is given by the quotient ring in the preceding definition, it is easy to see that if $(c_0, c_1, \ldots, c_{n-2}, c_{n-1})$ is a codeword in the skew cyclic code, then the skew-cyclic shift, $$\left(\sigma(c_{n-1}), \sigma(c_0), \sigma(c_1), \ldots, \sigma(c_{n-2})\right),$$ must also be a codeword.
There are many aspects of skew cyclic codes that we leave unaddressed in this article; the interested reader may see Chapter~8 of the {\em Concise Encyclopedia of Coding Theory}, \cite{Gluesing-Luerssen:2021}.

A natural generalization of skew cyclic codes are skew quasi-cyclic codes, which were introduced by Abualrub, Ghrayeb, Aydin, and Siap in \cite{Abualrub-et-al:2010}.

\begin{definition}
    Let $\F_{q^m}$ be a finite field, $\sigma$ an $\F_q$-automorphism of $\F_{q^m}$, and $m\mid n$. A subset $\C$ of $\F_{q^m}^{n\ell}$ is called a \emph{skew quasi-cyclic code} of length $n\ell$ and index $\ell$ if
    \begin{enumerate}
        \item[$(1)$] $\C$ is an $\F_{q^m}$-subspace of $\F_{q^m}^{n\ell}$, and
        \item[$(2)$] if $\bs{c} = (a_{0,0}, a_{0,1}, \ldots, a_{0,\ell-1}, a_{1,0}, \ldots, a_{1,\ell-1}, \ldots, a_{n-1,\ell-1})\in \C$, then the word
              \[T_{\ell,\sigma}(\bs{c})\! =\! (\sigma(a_{n-1,0}), \ldots,\! \sigma(a_{n-1,\ell-1}),\! \sigma(a_{0,0}), \ldots,\! \sigma(a_{0,\ell-1}), \ldots,\! \sigma(a_{n-2,\ell-1})) \]
              belongs to $\C$.
    \end{enumerate}
\end{definition}
Property (2) above is more simply stated as that the codewords of a skew quasi-cyclic code are closed under a ``skew-shift'' of index $\ell$. The authors assume $m \mid n$ so that $R_n = \F_{q^m}[X; \sigma]/(X^n - 1)$ is a ring, in which case we have the following characterization.

\begin{proposition}(\cite{Abualrub-et-al:2010}, Theorem 5)
    Say that $R_n = \F_{q^m}[X; \sigma]/(X^n - 1)$, and the parameters $n$ and $\ell$ are as above. Then the skew quasi-cyclic codes are in one-to-one correspondence with the $R_n$-submodules of $R_n^\ell$.
\end{proposition}

\subsection{\texorpdfstring{$\F_{q}$-Linear $\F_{q^m}$-Codes}{Fq-Linear Fqm Codes}}
Finally, the present article is concerned with generalizations of the above families of codes by considering subsets of $\F_{q^m}^n$ which fail to be linear subspaces over $\F_{q^m}$, but which are still linear over a subfield $\F_q$ of $\F_{q^m}$. Such codes were introduced and studied by Huffman \cite{Huffman:2010, Huffman:2013} and further studied by \cite{Sharma-Kaur:2017}. In this article, when the alphabet, $\F_{q^m}$, and scalar field $\F_q$ are understood, we will default to referring to these codes as \emph{nonlinear codes}.

Nonlinear codes are similar to, but distinct from, \emph{additive codes} which were first studied by Delsarte in \cite{delsarte1971}. Specifically, an additive code $\mathcal A$ over a field $\F_{q^m}$ is a subset of $\F_{q^m}^n$ which forms an abelian group under vector addition; see \cite{Bhunia2025} for a modern survey of the theory of additive codes.

\begin{remark}
    Certainly, any nonlinear code is additive, and, given a prime $p$, any additive code over $\F_{p^m}$ is an $\F_{p}$-linear $\F_{p^m}$-code.
    Thus, one may rather refer to nonlinear codes as \emph{sublinear codes} to accurately reflect their positioning between additive and linear codes.
\end{remark}

\section{Cyclic Structure of Additive Reed-Solomon Codes}\label{sec:nrs}
The genesis of the nonlinear skew quasi-cyclic codes investigated in the present article came from a study of nonlinear analogues of Reed-Solomon codes. Recently, Yadav and Sharma \cite{Yadav-2024} introduced these codes under the moniker of \emph{additive Reed-Solomon codes}, and we will quickly recall some relevant results from their article.

\begin{definition}
    Let $\F_{q^m}$ be a field extension of $\F_q$, $\boldsymbol{\alpha}=\left(\alpha_1, \alpha_2, \ldots, \alpha_n\right) \in \F_{q^m}^n$ a vector of distinct evaluation points $\alpha_i$, and let $k>0$. The \emph{nonlinear Reed-Solomon code}, $\nRS_q(\boldsymbol{\alpha},k)$, of length $n$ and parameter $k$ is given by
    \[\{\left(f(\alpha_1), f(\alpha_2), \ldots, f(\alpha_n) \right) : f(X) \in \F_q[X] \text{ with } \deg(f) < k\}.\]
\end{definition}

\begin{remark}
    Given $f\in \F_q[X]$ with $\deg(f) < k$ and $\boldsymbol{\alpha}$ an evaluation vector, consider the minimal polynomial $\min(\alpha_i; \F_q)$ of each $\alpha_i$ over $\F_q$. Perform long division of $f(X)$ by $\min(\alpha_i; \F_q)$ and evaluate the residue at $\alpha_i$. The resulting codeword equals the one obtained in the preceding definition. This strategy of defining an evaluation code is used in \cite{prcyu}.
\end{remark}

Note that the generator matrix for $\nRS_q(\bs \alpha, k)$ is identical to the generator matrix of the conventional linear Reed-Solomon code $\operatorname{RS}_q(\bs \alpha, k)$. However, we only consider the $\F_q$-linear combinations of the rows, which are vectors of length $n$ in $\F_{q^m}$.

\begin{remark}
    As we assume distinct values $\alpha_i$, the length $n$ of the code $\nRS_q(\boldsymbol{\alpha},k)$ is bounded, $n\leq q^m$. The parameter $k$ will coincide with the $\F_q$-dimension of the code $\nRS_q(\boldsymbol{\alpha},k)$ only for sufficiently small values of $k$, see \cite[Section 4]{Yadav-2024} for a detailed discussion and some upper bounds on $k$. Moreover, the code $\nRS_q(\bs \alpha, k)$ satisfies a weak version of the Singleton bound:
    \[
        d \leq n - \left\lceil \frac{k}{m} \right\rceil + 1.
    \]
    The $\F_q$-linear $\F_{q^m}$-codes which achieve this bound have variously been called \emph{MDS}, \emph{fractional MDS}, or \emph{quasi MDS} in the literature \cite{Ball:2023, Ball:2025, Huffman:2013, Martinez:2024}. See \cite[Section 4]{Yadav-2024} for a proof of the Singleton bound and for necessary and sufficient criteria for $\nRS_q(\bs \alpha, k)$ to be MDS.
\end{remark}

The authors of the present article noted that certain Reed-Solomon codes, notably those whose evaluation vector $\boldsymbol{\alpha}$ consists of the nonzero values of $\F_q$, which we will denote by $\F_q^*$, are BCH codes and optimal with respect to the Singleton bound. These BCH codes are commonly dubbed the ``narrow-sense Reed-Solomon codes'' with designed distance $\delta$, see \cite{Macwilliams-Sloane:1977} for a complete overview of this connection. A more succinct summary of this connection is: ``Full-length Reed-Solomon codes (i.e., $n = q-1$) are cyclic.'' Then there is a natural question for additive Reed-Solomon codes which motivated the present work:

\begin{center}
    \vspace{2pt}
    \boxed{\begin{minipage}{0.67\textwidth}\vspace{2pt}
            Are ``full-length" additive Reed-Solomon codes ``cyclic" in some sense?
            \vspace{2pt}
        \end{minipage}}
    \vspace{2pt}
\end{center}

Consider the problem of factoring a polynomial $f(X) \in \F_q[X]$ over an extension field $\F_{q^m}$. For proofs of the results mentioned in the following discussion the reader is encouraged to consult \cite{1997-lidl-nied}. Let $\sigma: \F_{q^m} \to \F_{q^m}$ be the Frobenius automorphism defined by $\sigma(X) = X^q$ which fixes $\F_q$. Given some $a \in \F_{q^m}$ such that $f(a) = 0$, it is simple to see that the collection $\{a, \sigma(a), \sigma^2(a), \ldots\}$ consists of roots of $f(X)$. Hence we will dub this (necessarily finite) collection the set of \emph{Galois conjugates} of $a$. If $a\ne 0$ and $\zeta$ is a primitive element of $\F_{q^m}$, write $a = \zeta^s$, and then the above family of roots may be written as $\{\zeta^s, \zeta^{qs}, \zeta^{q^2s}, \ldots\}$. The collection of the exponents is called the \emph{cyclotomic coset} of $s$, mod $q^m - 1$, typically denoted
\[C_s = \{sq^i \mod (q^m - 1) : i \geq 0\}.\]

It can be seen that the polynomial $g(X)$ defined by the product
\[g(X) := \prod_{i \in C_s} (X-\zeta^i)\] is the minimal polynomial $\min(a; \F_q)$ of $a = \zeta^s$ whose coefficients come from $\F_q$. More germane to evaluation codes, we have the following lemma whose proof is obvious.
\begin{lemma}
    Let $f(X) \in \F_q[X]$, and suppose that $\sigma$ is the Frobenius automorphism of $\F_{q^m}$ which fixes the subfield $\F_q$. Then for any $a \in \F_{q^m}$,
    \[\sigma(f(a)) = f(\sigma(a)).\]
\end{lemma}

If the evaluation set for the additive Reed-Solomon code is $\boldsymbol{\alpha} = (\zeta^i : i \in C_s) =(\zeta^s, \zeta^{qs}, \zeta^{q^2s}, \ldots)$ for some cyclotomic coset $C_s$ of size $|C_s| = n_s$, it can be seen that the individual codewords in $\nRS_{q}(\boldsymbol{\alpha}, k)$ are fixed under skew cyclic shifts; in other words, $(c_1, c_2, \ldots, c_{n_s}) = (\sigma(c_{n_s}), \sigma(c_1), \ldots, \sigma(c_{n_s -1}))$ for every $(c_1, c_2, \ldots, c_{n_s}) \in \nRS_q(\boldsymbol{\alpha}, k)$.

This insight motivates the following definition of codes which are invariant under skew shifts.
\begin{definition}\label{def:skewversions}
    Let $\F_q\leq \F_{q^a} \leq \F_{q^m}$ be field extensions and $\sigma: \F_{q^m}\rightarrow \F_{q^m}$ an $\F_q$-automorphism. An \emph{$\F_{q^a}$-linear skew cyclic $\F_{q^m}$-code} is an $\F_{q^a}$-subspace, $\C$, of $\F_{q^m}^n$ with the property that
    \begin{equation*}
        (a_0,\ldots,a_{n-1})\in \C \mbox{ implies } (\sigma(a_{n-1}),\sigma(a_0),\ldots,\sigma(a_{n-2}))\in \C.
    \end{equation*}

    A code $\C$ of length $n\ell$ with $m\mid n$ is called an \emph{$\F_{q^a}$-linear skew quasi-cyclic $\F_{q^m}$-code} of length $n\ell$ and index $\ell$ (for short, $(\sigma,\ell)$-QC code) if it is an $\F_{q^a}$-subspace of $\F_{q^m}^{n\ell}$, and it is invariant under a shift $s$ of codewords by $\ell$ positions followed by an application of $\sigma$ coordinatewise. In other words,
    \[\sigma(s(\boldsymbol c))\! =\!\!
        \begin{pmatrix}
            \sigma(c_{(n-1)\ell}) & \cdots & \sigma(c_{n\ell-1})     \\
            \sigma(c_{0})         & \cdots & \sigma(c_{\ell-1})      \\
            \vdots                & \ddots & \vdots                  \\
            \sigma(c_{(n-2)\ell}) & \cdots & \sigma(c_{(n-1)\ell-1})
        \end{pmatrix}\! \in\! \C\! \mbox{ for all}\!
        \begin{pmatrix}
            c_{0}         & \cdots & c_{\ell-1}  \\
            c_{\ell}      & \cdots & c_{2\ell-1} \\
            \vdots        & \ddots & \vdots      \\
            c_{(n-1)\ell} & \cdots & c_{n\ell-1}
        \end{pmatrix}\! \in \C. \]
\end{definition}

\begin{remark}
    These $\F_{q^a}$-linear skew quasi-cyclic $\F_{q^m}$-codes are related to codes studied extensively in the literature. For example, if $a=m$ and  $\sigma=\operatorname{id}$ we obtain quasi-cyclic codes, studied in \cite{Ling-Sole:2001}. For $a=m$ and $\sigma\ne \operatorname{id}$ we obtain linear skew quasi-cyclic codes, that were introduced in \cite{Abualrub-et-al:2010}.
\end{remark}

The above discussion gives an answer for our motivating question.

\begin{theorem}\label{thm-nrshasorbits}
    Suppose that we have an evaluation vector $\boldsymbol{\alpha} = (\alpha_1, \alpha_2, \ldots, \alpha_{m\ell}) \in \F_{q^m}^{m \ell}$ constructed by concatenating $\ell$ sets of Galois conjugates each of size m. Then the nonlinear Reed-Solomon code
    \[\nRS_q(\boldsymbol{\alpha},k) = \{\left(f(\alpha_1), f(\alpha_2), \ldots, f(\alpha_{m \ell}) \right) : f(X) \in \F_q[X] \text{ with } \deg(f) < k\}\]
    is an $\F_{q}$-linear skew quasi-cyclic $\F_{q^m}$-code of index $\ell$.
\end{theorem}

\section{Nonlinear Skew Cyclic and Nonlinear Skew Quasi-Cyclic Codes as Polynomials}\label{sec:polycharacterization}

Consider the field extensions $\F_q \leq \F_{q^a}\leq \F_{q^m}$, $\sigma$ the Frobenius automorphism with $\sigma(\alpha)=\alpha^q$, and $n$ a multiple of $m$. Notice that $(X^n-1)$, the left-sided ideal generated by $X^n-1$, is a two-sided ideal of $\F_{q^a}[X;\sigma]$ and $\F_{q^m}[X;\sigma]$, respectively. Then we will denote,

\[R_n := \F_{q^m}[X;\sigma]/\left( X^n-1\right),\]
\[ P_n:=\F_{q^a}[X;\sigma]/\left( X^n-1\right).\]

\begin{remark}
    Similar to the group algebras $R_n^{(q)}$ introduced in Huffman's study of $\F_q$-linear $\F_{q^m}$-cyclic codes, \cite{Huffman:2010}, we can express $R_n$ and $P_n$ as skew group rings
    \[R_n= \F_{q^m}[\mathbb{Z}_n, \psi], \qquad P_n= \F_{q^a}[\mathbb{Z}_n, \psi], \]
    where $\psi: \mathbb{Z}_n \rightarrow \operatorname{Aut}(\F_{q^m})$ maps $k\in \mathbb{Z}_n $ to $\sigma^k$. Then $P_n$ is canonically embedded into $R_n$.
\end{remark}

If $a=1$, we have that $\sigma$ fixes $\F_{q^a}=\F_q$. Then $\F_{q}[X;\sigma]\simeq \F_q[X]$ is commutative under polynomial multiplication and
\[ P_n=\F_{q}[X;\sigma]/\left( X^n-1\right)\simeq \F_{q}[X]/\left( X^n-1\right).\]

Let $\boldsymbol{c}=\left( c_{i,j}\right)_{0,0}^{n-1,\ell-1}\in \F_{q^m}^{n\ell}$. Define a map $\phi: \F_{q^m}^{n\ell}\rightarrow R_n^\ell$ by
\[\phi(\boldsymbol{c})=\left(c_0(X),\ldots,c_{\ell-1}(X)\right),\]
where each column $j$ of $\boldsymbol{c}$ determines a polynomial
\[c_{j}(X)=\sum_{i=0}^{n-1} c_{i,j}X^i \in R_n.\]

The map $\phi$ gives an isomorphism of the $\F_{q^m}$-vector spaces $\F_{q^m}^{n\ell}$ and $R_n^\ell$. We can use $\phi$ to associate every $(\sigma,\ell)$-QC code with a $P_n$-module structure.
\begin{theorem}\label{thm:skewQCiffPnsubmodule}
    An $\F_{q^a}$-subspace $\C\subseteq \F_{q^{m}}^{n\ell}$ is a $(\sigma,\ell)$-QC code of length $n\ell$ and index $\ell$ if and only if $\phi(\C)$ is a left $P_n$-submodule of the ring $R_n^\ell$.
\end{theorem}

\begin{proof} Let $\C$ be a $(\sigma,\ell)$-QC code of length $n\ell$ and index $\ell$. We claim that $\phi(\C)$ is a $P_n$-submodule of $R_n^\ell$. Evidently, $\phi(\C)$ is closed under addition, and we are going to check that $\phi(\C)$ is also closed under scalar multiplication by elements of $P_n$. Let $\phi(\boldsymbol{c})=\left(c_0(X),\ldots,c_{\ell-1}(X)\right)$. Given $j\in \{0,\ldots,\ell-1\}$ we have that
    \[Xc_j(X)=X\sum_{i=0}^{n-1} c_{i,j}X^i=\sum_{i=0}^{n-1} \sigma(c_{i,j})X^{i+1}\equiv \sum_{i=0}^{n-1} \sigma(c_{i-1,j})X^{i} \mod{(X^n-1)},\]
    hence
    \begin{align*}
        X \phi(\textbf{c}) & =\left(Xc_0(X),\ldots,Xc_{\ell-1}(X)\right)                                                                 \\
                           & =\left(\sum_{i=0}^{n-1} \sigma(c_{i-1,0})X^{i}, \ldots, \sum_{i=0}^{n-1} \sigma(c_{i-1,\ell-1})X^{i}\right) \\
                           & =\phi \begin{pmatrix}
                                       \sigma(c_{n-1,0}) & \cdots & \sigma(c_{n-1,\ell-1}) \\
                                       \sigma(c_{0,0})   & \cdots & \sigma(c_{0,\ell-1})   \\
                                       \vdots            & \ddots & \vdots                 \\
                                       \sigma(c_{n-2,0}) & \cdots & \sigma(c_{n-2,\ell-1})
                                   \end{pmatrix}\in \phi(\C).
    \end{align*}

    Then, by $\F_{q^a}$-linearity, $p(X)\phi(\boldsymbol{c})\in \phi(\C)$ for any $p(X)\in P_n$. Hence $\phi(\C)$ is a left $P_n$-submodule of $R_n^\ell$.

    Conversely, suppose that $D$ is a left $P_n$-submodule of $R_n^\ell$. Let $\C=\phi^{-1}(D)\subseteq \F_{q^{m}}^{n\ell}$. We claim that $\C$ is a $(\sigma,\ell)$-QC code of length $n\ell$ and index $\ell$.

    Since $\phi$ is an $\F_{q^a}$-vector space isomorphism, $\C$ is an $\F_{q^a}$-linear $\F_{q^m}$-code. To show that $\mathcal C$ is closed under $\sigma\circ s$, cf. Definition \ref{def:skewversions}, let $\boldsymbol{c}=\left( c_{i,j}\right)_{0,0}^{n-1,\ell-1}$ denote the preimage of the polynomial vector $(c_0(X),\ldots,c_{\ell-1}(X))\in D$. It is enough to show that $\sigma(s(\boldsymbol{c}))\in \phi^{-1}(D)$; i.e., $\phi\left(\sigma(s(\boldsymbol{c}))\right)\in D$. Notice that
    \begin{align*}
        \phi\left(\sigma(s(\boldsymbol{c}))\right)
         & =\phi\begin{pmatrix}
                    \sigma(c_{n-1,0}) & \cdots & \sigma(c_{n-1,\ell-1}) \\
                    \sigma(c_{0,0})   & \cdots & \sigma(c_{0,\ell-1})   \\
                    \vdots            & \ddots & \vdots                 \\
                    \sigma(c_{n-2,0}) & \cdots & \sigma(c_{n-2,\ell-1})
                \end{pmatrix}              \\
         & = \left(Xc_0(X),\ldots,Xc_{\ell-1}(X)\right)                          \\
         & = X\left(c_0(X),\ldots,c_{\ell-1}(X)\right)\in X \cdot D \subseteq D.
    \end{align*}

    Therefore, $\C$ is a $(\sigma,\ell)$-QC code of length $n\ell$ and index $\ell$.
\end{proof}

\begin{remark}
    For $\ell=1$ and $a=1$, Theorem \ref{thm:skewQCiffPnsubmodule} means that an $\F_q$-subspace $\C\subseteq \F_{q^{m}}^{n}$ is an $\F_q$-linear skew cyclic code if and only if $\phi(\C)$ is a left $P_n$-submodule of the ring $R_n$. If in addition $n=m$, an equivalent description is given by the study of circulant matrices (see also \cite[Section 8.3]{conciseencyclopedia}):

    A matrix $A=(a_{ij})\in \operatorname{Mat}_{m\times m}(\F_q)$ is called \emph{circulant} if $a_{i+1,j+1}=a_{ij}$ and $a_{i+1,1}=a_{i,m}$, for all $i,j\in \{1,\ldots,m-1\}$. Denote by $\operatorname{Circ}_m(\F_q)$ the collection of all circulant matrices in $\operatorname{Mat}_{m\times m}(\F_q)$. $\operatorname{Circ}_m(\F_q)$ is an $\F_q$-vector space of dimension $m$ and an $\F_q$-algebra.

    Given $g+(X^m-1)\in R_m$ with $g(X)=\sum_{i=0}^{m-1} g_i X^i$, let $\tilde{g}(X)=\sum_{i=0}^{m-1} g_i X^{q^i}$, its associated linearized polynomial. Let $\{\alpha, \alpha^q \ldots, \alpha^{q^{m-1}}\}$ be a normal basis of the extension $\F_q \leq \F_{q^m}$ and let $M_{\tilde{g}}$ denote the matrix of the $\F_q$-linear polynomial $\tilde{g}$ with respect to this basis. Set $\varphi(g+(X^m-1))=M_{\tilde{g}}$. Then $\varphi$ is a ring isomorphism between $R_m$ and $\operatorname{Mat}_{m\times m}(\F_q)$.

    We have that $\varphi\left( P_m\right)=\operatorname{Circ}_m(\F_q)$. Thus, if an $\F_q$-subspace $\C\subseteq \F_{q^{m}}^{m}$ is an $\F_q$-linear skew cyclic code, then $\varphi(\C)$ is a left $\operatorname{Circ}_m(\F_q)$-submodule of $\operatorname{Mat}_{m\times m}(\F_q)$.
\end{remark}

\section{Module Structure of \texorpdfstring{$\F_{q^a}[X;\sigma]$}{Skew-Polynomial}-Modules}\label{sec:structure1}

Consider the field extensions $\F_q \leq \F_{q^a}\leq \F_{q^m}$ with $r:=\frac{m}{a}$, the Frobenius automorphism $\sigma$ with $\sigma(\alpha)=\alpha^q$, and two $\F_{q^a}[X;\sigma]$-modules $C$ and $D$ such that $D\leq C\leq \F_{q^m}[X;\sigma]^\ell$. Note that $C$ and $D$ are free as $\F_{q^a}[X;\sigma]$-submodules of a free $\F_{q^a}[X;\sigma]$-module. Recall that $\F_{q^m}$ is an $\F_{q^a}$-vector space.  We can choose a normal basis $b_1, \ldots, b_r$ of $\F_{q^m}$ over $\F_{q^a}$, i.e., there exists  $\alpha\in \F_{q^m}\setminus \{0\}$ such that $b_1= \alpha, b_2=\alpha^{q^a}, \ldots, b_r = \alpha^{q^{a(r-1)}}$. In fact, \cite[Satz 1.2]{BLESSENOHL-KARSTENJOHNSEN:1986} asserts that such an $\alpha$ can be chosen to simultaneously generate a normal basis of $\F_{q^m}$ over any field intermediate to $\F_q\leq \F_{q^m}$.

Let $f(X)= \sum_{i} a_i X^i \in \F_{q^m}[X;\sigma]$. Since $\sigma$ is an $\F_q$-automorphism, $\sigma^{-1}$ is another $\F_q$-automorphism, and we can express
\[\sigma^{-i}(a_i)=\sum_{j=1}^r k_{ij} b_j= \sum_{j=1}^r k_{ij} \alpha^{q^{a(j-1)}} \mbox{ for suitable }k_{ij}\in \F_{q^a}.\]

Thus,

\begin{eqnarray*}
    f(X) = \sum_i a_i X^i = \sum_i X^i \sigma^{-i}(a_i) & = & \sum_i X^i\left( \sum_{j=1}^r k_{ij} b_j \right) \\ & = & \sum_{j=1}^r \left( \sum_i \sigma^i(k_{ij}) X^i \right)  b_j.
\end{eqnarray*}
If we consider $\F_{q^m}[X;\sigma]$ as an $\F_{q^a}[X;\sigma]$-module, it is now easy to see
\begin{equation}\label{Fqa-decomposition}
    \F_{q^m}[X;\sigma] = \bigoplus_{i=1}^r \F_{q^a}[X;\sigma] b_i.
\end{equation}

To continue with our discussion, let $\left\{\boldsymbol{r}_1, \ldots, \boldsymbol{r}_k\right\}$ be a generating set of $C$ as an $\F_{q^a}[X;\sigma]$-module, where

\[\boldsymbol{r}_i = \left( \boldsymbol{\tilde{r}}_{i,1} , \ldots, \boldsymbol{\tilde{r}}_{i,\ell} \right)\in \F_{q^m}[X;\sigma]^\ell.\]

Under the decomposition \eqref{Fqa-decomposition}, write each coordinate as

\[\boldsymbol{\tilde{r}}_{i,\iota}=\left( r_{i,(\iota-1)r+1}, r_{i,(\iota-1)r+2}, \ldots, r_{i,\iota r} \right) \in \F_{q^a} [X;\sigma]^{r}.\]

Then, $C= \langle \left( r_{i,1}, r_{i,2}, \ldots, r_{i,r\ell} \right): 1\le i\le k\rangle$ as an $\F_{q^a}[X;\sigma]$-module, and the rows of

\[ M= \begin{pmatrix}
        r_{1,1} & r_{1,2} & \cdots & r_{1,r\ell} \\
        r_{2,1} & r_{2,2} & \cdots & r_{2,r\ell} \\
        \vdots  & \vdots  & \ddots & \vdots      \\
        r_{k,1} & r_{k,2} & \cdots & r_{k,r\ell}
    \end{pmatrix}\]
generate $C$ as an $\F_{q^a}[X;\sigma]$-module. \\

For the next step in our discussion, we need to review total divisors and the Smith normal form for noncommutative rings. For reference, see \cite{Berrick:2000}.

\begin{definition}
    Let $s, t\in R$. Recall that $s$  is a {\it left divisor} of $t$ if there exists $u\in R$ such that $t= su$, and $s$  is a {\it right divisor} of $t$ if there exists $u\in R$ such that $t= us$. Left and right divisors will be denoted by $s \mid_l t$ and $s \mid_r t$, respectively. A much stronger statement is that $s$ is a {\it total divisor} of $t$. We say that $s$ totally divides $t$ (or $s \mid \mid t$) if for all $u\in R$, $s \mid_l ut$ and $s\mid_r tu$.
\end{definition}

\begin{theorem}[Smith Normal Form]
    Let $R$ be a (two-sided) Euclidean ring without zero divisors and let $A=(a_{ij})$ be an $n\times p$ matrix over $R$. Then, there exist invertible matrices $P$ and $Q$ over $R$ such that
    \[PAQ= \operatorname{diag}(e_1, \ldots, e_r, 0, \ldots, 0),\]
    a diagonal matrix with $e_1 \mid \mid e_2 \mid \mid \ldots \mid \mid e_{r}$ and $r\le \min\{n,p\}$. Moreover, $P$ may be obtained by applying a sequence of row operations to the $n\times n$ identity matrix, and $Q$ may be obtained by applying a sequence of column operations to the $p\times p$ identity matrix.
\end{theorem}
\begin{theorem}[The Invariant Factor Theorem]
    Let $M$ be a finitely generated right $R$-module, with $R$ a (two-sided) Euclidean ring without zero divisors. Then,
    \[ M \simeq R/e_1 R \oplus \ldots \oplus R/e_{r} R \oplus R^{s},\]
    where $e_1, \ldots, e_r$ are nonunits in $R$, $e_1 \mid \mid e_2 \mid \mid \ldots \mid \mid e_r$ and $s\geq0$.
\end{theorem}
We use the Smith normal form of the matrix $M$ to find a basis $\boldsymbol{c}_1, \boldsymbol{c}_2,\ldots, \boldsymbol{c}_{\xi}$ of $C$ as a free $\F_{q^a}[X;\sigma]$-module: the Smith normal form $J$ of $M$ provides us with invertible matrices $S$ and $T$ such that
\[J:= SMT = \begin{pmatrix}
        h_1    & 0      & \cdots & 0      & 0      & \cdots \\
        0      & h_2    & \cdots & 0      & 0      & \cdots \\
        \vdots & \vdots & \ddots & \vdots & \vdots & \cdots \\
        0      & 0      & \cdots & h_\xi  & 0      & \cdots \\
        0      & 0      & \cdots & 0      & 0      & \cdots \\
        \vdots & \vdots & \vdots & \vdots & \vdots & \ddots \\
    \end{pmatrix}, \]
where $h_1, h_2, \ldots, h_\xi$ are the nonzero diagonal entries of $J$ with $h_i$ a total divisor of $h_{i+1}$ for all $i=1, \ldots, \xi-1$. The basis elements $\boldsymbol{c}_i$ are the nonzero rows of $JT^{-1}$, while the rows of $T^{-1}$ form an associated basis of $\F_{q^a}[X;\sigma]^{r\ell}$.

Next, let $\left\{\boldsymbol{t}_1, \ldots, \boldsymbol{t}_\delta\right\}$ be a generating set of $D$ as an $\F_{q^a}[X;\sigma]$-module, where
\[\boldsymbol{t}_i = \left( \boldsymbol{\tilde{t}}_{i,1} , \ldots, \boldsymbol{\tilde{t}}_{i,\ell} \right)\in \F_{q^m}[X;\sigma]^\ell.  \]
Under the decomposition \eqref{Fqa-decomposition}, write
\[\boldsymbol{\tilde{t}}_{i,\iota}=\left( t_{i,(\iota-1)r+1}, t_{i,(\iota-1)r+2}, \ldots, t_{i,\iota r} \right) \in \F_{q^a} [X;\sigma]^{r}. \]
Then, $D= \langle \left( t_{i,1}, t_{i,2}, \ldots, t_{i,r\ell} \right): 1\le i \le \delta \rangle$ as an $\F_{q^a}[X;\sigma]$-module.
Express each row vector $\left( t_{i,1}, \ldots, t_{i,r\ell} \right)$ with respect to the basis $\boldsymbol{c}_1, \boldsymbol{c}_2,\ldots, \boldsymbol{c}_{\xi}$ of $C$ as
\[ \left( t_{i,1}, \ldots, t_{i,r\ell} \right) = \sum_{j=1}^\xi g_{i,j} \boldsymbol{c}_j, \; g_{i,j}  \in\F_{q^a} [X;\sigma]. \]
Then the rows of
\[ N= \begin{pmatrix}
        g_{1,1}      & g_{1,2}      & \cdots & g_{1,\xi}      \\
        g_{2,1}      & g_{2,2}      & \cdots & g_{2,\xi}      \\
        \vdots       & \vdots       & \ddots & \vdots         \\
        g_{\delta,1} & g_{\delta,2} & \cdots & g_{\delta,\xi}
    \end{pmatrix}\]
generate $D$ as an $\F_{q^a}[X;\sigma]$-submodule of $C$, with respect to the basis vectors $\boldsymbol{c}_1, \boldsymbol{c}_2,\ldots, \boldsymbol{c}_{\xi}$ of~$C$.
Finally, put the matrix $N$ into Smith normal form with nonzero diagonal entries $d_1, d_2, \ldots, d_{\zeta}$,
\[  \begin{pmatrix}
        d_1    & 0      & \cdots & 0       & 0      & \cdots \\
        0      & d_2    & \cdots & 0       & 0      & \cdots \\
        \vdots & \vdots & \ddots & \vdots  & \vdots & \cdots \\
        0      & 0      & \cdots & d_\zeta & 0      & \cdots \\
        0      & 0      & \cdots & 0       & 0      & \cdots \\
        \vdots & \vdots & \vdots & \vdots  & \vdots & \ddots \\
    \end{pmatrix}.\]

This gives a new basis $\boldsymbol{c}_1^*, \boldsymbol{c}_2^*, \ldots, \boldsymbol{c}_{\xi}^*$ for $C$ such that $ \left\{  d_i \boldsymbol{c}_i^* : i =1, \ldots, \zeta  \right\} $ is a basis of~$D$. Set $d_i=0$ for all $i>\zeta$. We have that, as $\F_{q^a}[X;\sigma]$-modules,

\begin{align*}
    C / D & =\left(\bigoplus_{i=1}^\xi \F_{q^a} [X;\sigma]\boldsymbol{c}_i^* \right) \Big/ \left(\bigoplus_{i=1}^\xi  \F_{q^a} [X;\sigma]d_i\boldsymbol{c}_i^* \right) \\
          & \simeq \bigoplus_{i=1}^\xi \left(  \F_{q^a} [X;\sigma] /  \F_{q^a} [X;\sigma] d_i \right).
\end{align*}

\section{Total Divisors in Quotients of Skew Polynomial Rings}\label{sec:totaldivisors}

In Section \ref{sec:structure1} we described the module structure of $C/D$ using two applications  of the Smith normal form, where $C$ and $D$ are arbitrary $\F_{q^a}[X;\sigma]$-submodules such that $D\leq C\leq \F_{q^m}[X;\sigma]^\ell$. Nonlinear skew quasi-cyclic codes, the primary focus of this article, fall into the special case where $D= (X^n-1)\F_{q^m}[X;\sigma]^\ell$, the $\F_{q^m}[X;\sigma]$-submodule of all scalar multiples by the scalar $X^n-1\in \F_{q^a}[X;\sigma]$. In this special case, it is sufficient to use a single application of the Smith normal form. This reduction requires a rather deep result about total divisors which is presented in this section.

Let $R=\F_{q^a}[X;\sigma]$ for the Frobenius automorphism $\sigma$ with $\sigma(\alpha)=\alpha^q$. Recall that $R$ has center $Z(R)= \F_q[X^a; \sigma]$ and that $R$ is a left and right Euclidean ring without zero divisors. In particular, $R$ is both a left and right principal ideal ring.
For any $t \in R$ we will denote by $Rt$ the left ideal of $R$ generated by $t$, while $tR$ is the right ideal of $R$ generated by $t$, and
\[R t R=\left\{\sum_{i} r_{i} t r_{i}^{\prime} \mid r_{i}, r_{i}^{\prime} \in R\right\}\]
denotes the two-sided ideal of $R$ generated by $t$.

\begin{definition}
    For any left, right, or two-sided ideal $I\ne 0$ of $R$, denote by $g_I \in I$ the monic polynomial of smallest degree in $I$. This polynomial $g_I$ is unique and generates $I$. We can write:
    \begin{itemize}
        \item $I=Rg_I$, if $I$ is a left ideal,
        \item $I=g_IR$, if $I$ is a right ideal,
        \item $I=Rg_I=g_IR=Rg_IR$, if $I$ is a two-sided ideal.
    \end{itemize}
\end{definition}

We note some basic criteria for checking if a polynomial is a total divisor of another.
\begin{lemma}\label{thm:allb-fistcriteria}
    Let $s,t\in R$ with $s\ne 0$. The polynomial $s$ is a total divisor of $t$ if and only if $s\mid_l \alpha X^{k} t,s\mid_r t \alpha X^{k}$ for all $\alpha \in \mathbb{F}_{q^a}$ and $0\le k < \deg s$.
\end{lemma}
\begin{proof}
    We need to check that $s\mid_l g t$ and $s\mid_{r} t g$ for all $g \in R=\mathbb{F}_{q^a}[X; \sigma]$. If $\deg g \geq \deg s$, we can use left or right division algorithm of $g$ by $s$. In particular, we can write $g=q s+r$ with $q, r \in R$, $\deg r<\deg s$, and $s \mid_r t g$ if and only if $s\mid_{r} t r$. Thus, it suffices to check only those $g \in R$ with $\deg g<\deg s$. Moreover, by additivity, it suffices to check only $g=\alpha X^{k}$ with $\alpha \in \mathbb{F}_{q^a}, k < \deg s.$
\end{proof}

\begin{lemma}\label{thm:allb-secondcriteria}
    Let $s,t\in R$ with $t\ne 0$. The polynomial $s$ is a total divisor of $t$ if and only if $s\mid_l g_{I}$ and $s\mid_{r} g_{I}$, where $I=R t R$.
\end{lemma}

\begin{proof}
    Suppose that $s$ is a total divisor of $t$. Then $s$ is both a left and right divisor of every element of the form $\sum_i r_i t r_i^\prime$. That is, $s$ is both a left and right divisor of every element in $I =RtR\ne 0$. In particular, $s$ is both a left and right divisor of $g_I \in I$.

    Now, let $s\mid_l g_{I}$, $s\mid_{r} g_{I}$, and $g\in R$. Then $tg\in RtR=Rg_I$, which implies that $g_I \mid_r tg$, and since $s\mid_r g_I$, we have  $s \mid_r tg$ by transitivity. Similarly, we can show that $s\mid_l gt$.
\end{proof}

The following result is instrumental in the one-step construction of a basis for the $P_n$-submodules of $R_n^\ell$ in Section \ref{sec:structure2} and in the results below. A proof may be found in \cite[Thm. 1.1.22]{Jacobson:2009}.

\begin{lemma}\label{lem:selftotaldivisor}
    Let $s\in R$. Then $s$ is a total divisor of itself if and only if $s=\gamma c X^k$, where $\gamma \in \F_{q^a}$, $k \geq 0$, and $c\in Z(R)= \F_q[X^a; \sigma]$ a monic polynomial in the center of $R$.
\end{lemma}

\begin{corollary}\label{thm:form-of-gI}
    Let $I=RtR\ne 0$. Then $g_I$ is a total divisor of itself and $g_I= c X^k$ with $c$ a monic polynomial in $Z(R)= \F_q[X^a; \sigma]$ and $k\ge 0$.
\end{corollary}

\begin{proof}
    Notice that $g_I \mid_l g_I$ and $g_I \mid_r g_I$. Thus, by Lemma \ref{thm:allb-secondcriteria} it follows that $g_I$ is a total divisor of itself. Then, by Lemma \ref{lem:selftotaldivisor}, we have that $g_I = \gamma c X^k$, with $\gamma\in \F_{q^a},$ $k\ge 0$ and $c\in Z(R)$ monic. As $g_I$ is monic, we have that $\gamma=1$ and $g_I = c X^k$.
\end{proof}

\begin{remark}
    Let $I=RtR\ne 0$. Then $g_I=cX^k$, where $k$ denotes the maximum integer with $X^k \mid_r t$, so that $t=t'X^k$ for some $t'\in R$, and $c$ is the maximum degree monic polynomial in $Z(R)$ with $c\mid_r t^\prime$.
\end{remark}

\begin{proposition}\label{thm:intermediate}\
    \begin{enumerate}
        \item[$(1)$] Let $s,t\in R$ with $st\in Z(R)$. Then $st = ts$.
        \item[$(2)$] Let $s\in R, t\in Z(R)$ with $s\mid_r t$. Then $s\mid_l t$ and $s \mid\mid t$.
        \item[$(3)$] Let $s\in Z(R), t\in R$ with $s\mid_r t$. Then $s\mid_l t$ and $s \mid\mid t$.
        \item[$(4)$] Let $s,t\in R$ and $c\in Z(R)$ with $s\mid_r c$ and $c\mid_r t$. Then $s \mid\mid t$.
    \end{enumerate}
\end{proposition}

\begin{proof}
    $(1)$ For $s=0$, the statement is trivial. For $s\ne 0$, notice that $st\in Z(R)$ gives $s(st) = (st)s = s(ts)$ and use that $R$ has no zero divisors to cancel $s$.

    $(2)$ Notice that $s\mid_r t$ implies that $t=s^\prime s$ for some $s^\prime\in R$. Since $t\in Z(R)$ it follows that $t=s^\prime s=s s^\prime$, see Part $(1)$. Thus, $s\mid_l t$. Then, for every $g\in R$, we have that $tg=gt=gs^\prime s$, that is, $s\mid_r tg$. Similarly, we can show that $s\mid_l gt$, and $s$ is a total divisor of $t$.

    $(3)$ Observe that $s \mid_r t$ implies that $t=s^\prime s$ for some $s^\prime \in R$. Since $s\in Z(R)$ it follows that $t=s^\prime s= s s^\prime$. Thus, $s\mid_l t$. Then, for every $g\in R$, we have that $tg =s^\prime s g = s^\prime g s $, that is, $s \mid_r tg$. Similarly, we can show that $s\mid_l gt$, and $s$ is a total divisor of $t$.

    $(4)$ Since $c\in Z(R)$ and $s\mid_r c$, Part $(2)$ implies $s\mid \mid c$. Similarly, since $c\in Z(R)$ and $c\mid_r t$, Part $(3)$ implies $c\mid \mid t$. By transitivity, $s\mid\mid t$ follows.
\end{proof}

\begin{theorem}\label{thm:grcdmakesmagic}
    If $s$ is a total divisor of $t$ and $a\mid n$, then $\operatorname{gcrd}(s,X^n-1)$ is a total divisor of $\operatorname{gcrd}(t,X^n-1)$.
\end{theorem}

\begin{proof}
    If $t=0$, then $\operatorname{gcrd}(t,X^n-1) = X^n-1 \in Z(R)$. Thus, we conclude $\operatorname{gcrd}(s,X^n-1) \mid_r \operatorname{gcrd}(t,X^n-1)$, and $\operatorname{gcrd}(s,X^n-1)$ is a total divisor of $\operatorname{gcrd}(t,X^n-1)$ by Proposition \ref{thm:intermediate}(2).

    It remains to consider the case $t\ne 0$.
    Suppose that $s\mid\mid t$ and let $I = RtR\ne 0$. Then, $s\mid_r g_I$ by Lemma \ref{thm:allb-secondcriteria}. Since $t\in I=Rg_I$, we have that $g_I\mid_r t$. Then, $s\mid_r g_I$ and $g_I\mid_r t$ imply that $\operatorname{gcrd}(s,X^n-1)\mid_r \operatorname{gcrd}(g_I,X^n-1)$ and $\operatorname{gcrd}(g_I,X^n-1)\mid_r \operatorname{gcrd}(t,X^n-1)$. If we can show that $\operatorname{gcrd}(g_I,X^n-1)$ belongs to $Z(R)$, then the desired conclusion follows from Proposition \ref{thm:intermediate}(4).

    To show that $\operatorname{gcrd}(g_I,X^n-1)\in Z(R)$, recall that $g_I= c X^k$ with $c\in Z(R)$ a monic polynomial and $k\ge 0$, see Corollary \ref{thm:form-of-gI}.  We have that $\operatorname{gcrd}(g_I,X^n-1)=\operatorname{gcrd}(cX^k,X^n-1)=\operatorname{gcrd}(c,X^n-1)$, since the right Euclidean algorithm for calculating $\operatorname{gcrd}(cX^k, X^n-1)$ performs completely in the subring $\F_q[X;\sigma] \simeq\F_q[X]$ of $R$. Furthermore, $\operatorname{gcrd}(c,X^n-1)\in Z(R)$ because $c, X^n-1\in Z(R)$ and the right Euclidean algorithm fully performs in the subring $Z(R)=\F_q[X^a;\sigma] \simeq \F_q[X^a]$ of $R$. Hence, $\operatorname{gcrd}(g_I,X^n-1)\in Z(R)$.
\end{proof}

\section{Structure of \texorpdfstring{$P_n$}{Pn}-Submodules of \texorpdfstring{$R_n^\ell$}{Rnl}}\label{sec:structure2}
Consider the field extensions $\F_q \leq \F_{q^a}\leq \F_{q^m}$ and $n$ a multiple of $m$.
Again, let
\[R_n = \F_{q^m}[X;\sigma]/\left( X^n-1\right),\]
\[ P_n=\F_{q^a}[X;\sigma]/\left( X^n-1\right),\]
where $\sigma$ is the Frobenius automorphism given by $\sigma(\alpha)=\alpha^q$, for $\alpha \in \F_{q^m}$. Let $\C$ be an $\F_{q^a}$-linear skew quasi-cyclic $\F_{q^m}$-code of length $n\ell$ and index $\ell$ with $\dim_{\F_{q^a}}(\C) = k$. We have that $\C$ is a $P_n$-submodule of $R_n^\ell$. This section is devoted to studying the $P_n$-module structure of  $\C$.

Let $\pi: \F_{q^m}[X;\sigma]^\ell\rightarrow  R_n^\ell$ denote the canonical projection, i.e., identifying coordinatewise with the remainder in $\F_{q^m}[X;\sigma]$ modulo $(X^n-1)$. Then $\pi^{-1}\left(\C \right)$, the preimage of $\C$ under $\pi$, is an $\F_{q^a}[X;\sigma]$-module. The code $\C$ has a generator matrix of the form

\[ G= \begin{pmatrix}
        \boldsymbol{r}_1 &
        \boldsymbol{r}_2 &
        \ldots           &
        \boldsymbol{r}_k
    \end{pmatrix}^t\]
with $\boldsymbol{r}_i = \left( \boldsymbol{\tilde{r}}_{i,1} + (X^n-1), \ldots, \boldsymbol{\tilde{r}}_{i,\ell} + (X^n-1)\right)\in R_n^\ell$ and  $\boldsymbol{\tilde{r}}_{i,\iota} \in \F_{q^m}[X;\sigma]$,
where the rows $\boldsymbol{r}_1 ,\boldsymbol{r}_2,\ldots, \boldsymbol{r}_k$ generate $\C$ as a $P_n$-module.

As in Section 5, let $b_1, \ldots, b_r$ be a normal basis of $\F_{q^m}$ over $\F_{q^a}$, where $r=\frac{m}{a}$. If we consider $\F_{q^m}[X;\sigma]$ as an $\F_{q^a}[X;\sigma]$-module, we have as in \eqref{Fqa-decomposition} that
\begin{equation*}\label{Fqa2-decomposition}
    \F_{q^m}[X;\sigma] = \bigoplus_{i=1}^{r} \F_{q^a}[X;\sigma] b_{i }   .
\end{equation*}
Under this decomposition, write
\[\boldsymbol{\tilde{ {r}}}_{i,\iota}=\left( r_{i,(\iota-1)r+1}, r_{i,(\iota-1)r+2}, \ldots, r_{i,\iota r} \right) \in \F_{q^a} [X;\sigma]^{r}. \]

Let $e_i$ be the element in $\F_{q^m}[X;\sigma]^\ell$ with entry $1$ as the $i$-th coordinate and entry $0$ in all the other coordinates. Notice that, as a left $\F_{q^a}[X;\sigma]$-module, the $\F_{q^m}[X;\sigma]$-submodule  $(X^n-1)\F_{q^m}[X;\sigma]^\ell$ of all scalar multiples by the scalar $X^n-1\in \F_{q^a}[X;\sigma]$
has the basis
\begin{align*}
     & (X^n-1)b_1e_1, \ldots, (X^n-1)b_re_1,       \\
     & (X^n-1)b_1e_2, \ldots, (X^n-1)b_re_2,       \\
     & \qquad \qquad \qquad \ \ \vdots             \\
     & (X^n-1)b_1e_\ell, \ldots, (X^n-1)b_re_\ell.
\end{align*}
Then, \[\pi^{-1}(\C)= \left\langle  \left( r_{i,1}, r_{i,2}, \ldots, r_{i,r\ell} \right) : 1\le i\le k \right\rangle +(X^n-1)\F_{q^a}[X;\sigma]^{r\ell}\] as an $\F_{q^a}[X;\sigma]$-module, and the rows of
\begin{equation}\label{Smith1}
    M= \begin{pmatrix}
        r_{1,1} & r_{1,2} & \cdots & r_{1,r\ell} \\
        r_{2,1} & r_{2,2} & \cdots & r_{2,r\ell} \\
        \vdots  & \vdots  & \ddots & \vdots      \\
        r_{k,1} & r_{k,2} & \cdots & r_{k,r\ell} \\ \hline
        X^n-1   & 0       & \cdots & 0           \\
        0       & X^n-1   & \cdots & 0           \\
        \vdots  & \vdots  & \ddots & \vdots      \\
        0       & 0       & \cdots & X^n-1
    \end{pmatrix}
\end{equation}
generate $\pi^{-1}(\C)$ as an $\F_{q^a}[X;\sigma]$-module. As $\F_{q^m}[X;\sigma]$ is a free module over the one-sided Euclidean ring $\F_{q^a}[X;\sigma]$, $\pi^{-1}(\C)\subseteq \F_{q^m}[X;\sigma]^\ell$ is a free $\F_{q^a}[X;\sigma]$-module.

\begin{remark}
    From here, one may apply the double Smith normal form approach from Section \ref{sec:structure1} to the $\F_{q^a}[X;\sigma]$-modules $D\leq C\leq \F_{q^m}[X;\sigma]^\ell$ with $C=\pi^{-1}(\C)$ and $D=(X^n-1)\F_{q^m}[X;\sigma]^\ell$ to obtain the $P_n$-module structure of  $\C \simeq C/D$. In the following, we will present a computationally more efficient alternative method that requires only one application of the Smith normal form.
\end{remark}

We start with the Smith normal form $J$ of the matrix $M$ which provides us with invertible matrices $S$ and $T$ such that
\begin{equation} \label{Smith2}
    J := SMT =\begin{pmatrix}
        h_1    & 0      & \cdots & 0         \\
        0      & h_2    & \cdots & 0         \\
        \vdots & \vdots & \ddots & \vdots    \\
        0      & 0      & \cdots & h_{r\ell} \\
        0      & 0      & \cdots & 0         \\
        \vdots & \vdots & \ddots & \vdots    \\
        0      & 0      & \cdots & 0         \\
    \end{pmatrix}.
\end{equation}

Note that the polynomials $h_i$ must all be nonzero as the matrix $M$ has full rank. In fact, we have the following much stronger general division result.
\begin{proposition}
    The polynomials $h_i$ on the main diagonal of $J$ are total divisors of $X^n-1$ in $\F_{q^a}[X;\sigma]$. In particular, $h_i d_i = d_i h_i = X^n-1$ for suitable $d_i\in \F_{q^a}[X;\sigma]$.
\end{proposition}
\begin{proof}
    Denote
    \[A:= \begin{pmatrix}
            r_{1,1} & r_{1,2} & \cdots & r_{1,r\ell} \\
            r_{2,1} & r_{2,2} & \cdots & r_{2,r\ell} \\
            \vdots  & \vdots  & \ddots & \vdots      \\
            r_{k,1} & r_{k,2} & \cdots & r_{k,r\ell}
        \end{pmatrix} \quad \mbox{ and } \quad D := (X^n-1) \operatorname{Id}_{r\ell}. \]
    The Smith normal form of $A$ provides us with invertible matrices  $\overline{S}$ and $\overline{T}$ such that
    \[ \overline{J} :=\overline{S}A\overline{T}= \begin{pmatrix}
            a_1    & \cdots & 0        & 0      & \cdots & 0      \\
            \vdots & \ddots & \vdots   & \vdots & \cdots & \vdots \\
            0      & \cdots & a_\kappa & 0      & \cdots & 0      \\
            0      & \cdots & 0        & 0      & \cdots & 0      \\
            \vdots & \vdots & \vdots   & \vdots & \ddots & \vdots \\
        \end{pmatrix}  ,\]
    where $a_1, a_2, \ldots, a_\kappa$ with $\kappa \le r\ell$ are the nonzero diagonal entries of $\overline{J}$ with $a_i$ a total divisor of $a_{i+1}$ for all $i=1, \ldots, \kappa-1$. Then,
    \begin{align*}   & \quad\,\,
               \begin{pNiceArray}{c|c}
                       \overline{S} & \mathbf{0} \\
                       \hline\\[-4mm]
                       \mathbf{0} & \overline{T}^{-1} \\
                   \end{pNiceArray} \begin{pmatrix}
                                        A \\ \hline
                                        D
                                    \end{pmatrix}
               \overline{T} = \begin{pmatrix}
                                      \overline{S}A\overline{T} \\ \hline\\[-4mm]
                                      \overline{T}^{-1}D\overline{T}
                                  \end{pmatrix}  = \begin{pmatrix}
                                                       \overline{J} \\ \hline
                                                       D
                                                   \end{pmatrix}          \\
               \\
               = & \begin{pmatrix}
                           a_1    & \cdots & 0        & 0      & \cdots & 0      \\
                           \vdots & \ddots & \vdots   & \vdots & \cdots & \vdots \\
                           0      & \cdots & a_\kappa & 0      & \cdots & 0      \\
                           0      & \cdots & 0        & 0      & \cdots & 0      \\
                           \vdots & \vdots & \vdots   & \vdots & \ddots & \vdots \\ \hline
                           X^n-1  & \cdots & 0        & 0      & \cdots & 0      \\
                           \vdots & \ddots & \vdots   & \vdots & \cdots & \vdots \\
                           0      & \cdots & X^n-1    & 0      & \cdots & 0      \\
                           0      & \cdots & 0        & X^n-1  & \cdots & 0      \\
                           \vdots & \vdots & \vdots   & \vdots & \ddots & \vdots \\
                           0      & \cdots & 0        & 0      & \cdots & X^n-1
                       \end{pmatrix}.
    \end{align*}
    Applying row operations to the rows of this matrix, we can implement the Euclidean algorithm to find the greatest common right divisor of the respective entries $a_i$ and $X^n-1$ in the first $\kappa$ columns of this matrix. Row swaps then result in the following equivalent matrix:

    \begin{align}\label{snf:6.1}
        \begin{pmatrix}
            \operatorname{gcrd}(a_1, X^n-1) & \cdots & 0                                    & 0      & \cdots & 0      \\
            \vdots                          & \ddots & \vdots                               & \vdots & \cdots & \vdots \\
            0                               & \cdots & \operatorname{gcrd}(a_\kappa, X^n-1) & 0      & \cdots & 0      \\
            0                               & \cdots & 0                                    & X^n-1  & \cdots & 0      \\
            \vdots                          & \vdots & \vdots                               & \vdots & \ddots & \vdots \\
            0                               & \cdots & 0                                    & 0      & \cdots & X^n-1  \\
            0                               & \cdots & 0                                    & 0      & \cdots & 0      \\
            \vdots                          & \vdots & \vdots                               & \vdots & \vdots & \vdots \\
        \end{pmatrix}
    \end{align}
    Recall that $a_i \mid\mid a_{i+1}$ for all $i=1, \ldots, \kappa-1$. By Theorem~ \ref{thm:grcdmakesmagic}, $\operatorname{gcrd}(a_i, X^n-1)$ is a total divisor of $\operatorname{gcrd}(a_{i+1}, X^n-1)$, while $\operatorname{gcrd}(a_\kappa, X^n-1)\mid \mid X^n-1$ with Proposition \ref{thm:intermediate}(2) and $X^n-1\mid \mid X^n-1$ with Lemma \ref{lem:selftotaldivisor}. Thus, \eqref{snf:6.1} is the Smith normal form of $M$. Then, $h_i$ is either $\operatorname{gcrd}(a_{i}, X^n-1)$ or $X^n-1$. In either case, $h_{i}$ is a total divisor of $X^n-1$ by Proposition \ref{thm:intermediate}(2) and Lemma \ref{lem:selftotaldivisor}, respectively.

    Let $d_i\in \F_{q^a}[X;\sigma]$ with $h_i d_i = X^n-1$. Then $X^n-1= h_{i}d_{i}=d_{i}h_{i}$ since $X^n-1\in Z\left( \F_{q^a}[X;\sigma] \right)$, cf. Proposition \ref{thm:intermediate}(1).
\end{proof}

Let $\boldsymbol{q}_i$ denote the rows of $T^{-1}$ in \eqref{Smith2}. Let $\boldsymbol{c}_i := h_i \boldsymbol{q}_i$. By the Smith normal form we have that $\boldsymbol{q}_1, \ldots, \boldsymbol{q}_{r\ell}$ is a basis of $\F_{q^a}[X;\sigma]^{r\ell}$ and that $\boldsymbol{c}_1, \ldots, \boldsymbol{c}_{r\ell}$ is a basis of $\pi^{-1}(\C)$.

Let $d_i$ be the right quotient of $X^n-1$ by $h_i$. Since $X^n-1 \in Z\left( \F_{q^a}[X;\sigma] \right)$, by Proposition \ref{thm:intermediate}(1), we have that $h_i d_i = d_i h_i = X^n-1$. Moreover, $d_i \boldsymbol{c}_i = d_i h_i\boldsymbol{q}_i =(X^n-1)\boldsymbol{q}_i$, and $d_1 \boldsymbol{c}_1, \ldots, d_{r\ell} \boldsymbol{c}_{r\ell}$ is a basis of $(X^n-1)\F_{q^a}[X;\sigma]^{r\ell}$.

Observing that the resulting quotients $\F_{q^a} [X;\sigma] /  \F_{q^a} [X;\sigma]d_i$ naturally inherit the $P_n$-module structure of $\mathcal C$  and $R_n^\ell$, we have shown the main theorem of this article.

\begin{theorem}\label{thm:main_theorem}
    Let $\C$ be an $\F_{q^a}$-linear skew quasi-cyclic $\F_{q^m}$-code of length $n\ell$ and index $\ell$. Then, both as an $\F_{q^a}[X;\sigma]$-module and as a $P_n$-module,
    \begin{align*}
        \C & \simeq \pi^{-1}(\C) \Big/ \left((X^n-1)\F_{q^a}[X;\sigma]^{r\ell} \right)                                                                                      \\
           & =\left(\bigoplus_{i=1}^{r\ell} \F_{q^a} [X;\sigma]\boldsymbol{c}_i \right) \Big/ \left(\bigoplus_{i=1}^{r\ell}  \F_{q^a} [X;\sigma]d_i\boldsymbol{c}_i \right) \\
           & \simeq \bigoplus_{i=1}^{r\ell} \left(  \F_{q^a} [X;\sigma] \Big/  \F_{q^a} [X;\sigma]d_i  \right).
    \end{align*}
\end{theorem}
\begin{example}
    Let $q=2$, $a=1$, $n=m=4$. $\F_{16}$ has $\{\alpha^3, \alpha^6, \alpha^{12}, \alpha^{9}\}$ as basis over $\F_2$, where $\alpha$ is a primitive element that generates the multiplicative group $\F_{16}\setminus \{0\}$. With $\sigma(a)=a^2$, we have $\F_2[X;\sigma]= \F_2[X]$. Let $\C= P_4\cdot \left( \alpha^8 + \alpha^2X + \alpha^8 X^2 + \alpha^2 X^3\right)$, an $\F_2$-linear skew cyclic code. The generator of $\C$ can be expressed as
    \begin{align*}
        \alpha^8 + \alpha^2X + \alpha^8 X^2 + \alpha^2 X^3
        =(X+X^2)\alpha^3 + (X^2+X^3)\alpha^6 +(1+X^3)\alpha^{12} +(1+X)\alpha^9.
    \end{align*}
    As an $\F_2[X]$-module, $\pi^{-1}(\C)$ is generated by the rows of
    \[M= \begin{pmatrix}
            \scriptstyle{X+X^2} & \scriptstyle{X^2+X^3} & \scriptstyle{1+X^3} & \scriptstyle{1+X}   \\
            \scriptstyle{X^4-1} & 0                     & 0                   & 0                   \\
            0                   & \scriptstyle{X^4-1}   & 0                   & 0                   \\
            0                   & 0                     & \scriptstyle{X^4-1} & 0                   \\
            0                   & 0                     & 0                   & \scriptstyle{X^4-1} \\
        \end{pmatrix}.\]

    The Smith normal form provides us with invertible matrices $S$ and $T$ such that
    \[SMT = J= \begin{pmatrix}
            \scriptstyle{1+X} & 0                   & 0                   & 0                   \\
            0                 & \scriptstyle{X^4-1} & 0                   & 0                   \\
            0                 & 0                   & \scriptstyle{X^4-1} & 0                   \\
            0                 & 0                   & 0                   & \scriptstyle{X^4-1} \\
            0                 & 0                   & 0                   & 0
        \end{pmatrix}. \]

    That is, $h_1=1+X$, $h_2=h_3=h_4=X^4-1$. Therefore, $d_1=\frac{X^4-1}{X+1}=X^3 + X^2 + X + 1$, $d_2=d_3=d_4=1$ and
    \[ \C\simeq  \F_2 [X] /\, \F_2 [X] \left( X^3 + X^2 + X + 1 \right).  \]
\end{example}
Given the above decomposition, the natural question of an induced dual code for $\C$ arises. For this, we introduce an inner product associated with the basis vectors $\boldsymbol{q}_i$.

\begin{definition}
    Recall
    \[ T^{-1}= \begin{pmatrix}
            \boldsymbol{q}_1 \\
            \boldsymbol{q}_2 \\
            \vdots           \\
            \boldsymbol{q}_{r\ell}
        \end{pmatrix}.\]
    For row vectors $\boldsymbol{a},\boldsymbol{b}\in \F_{q^a}[X;\sigma]^{r\ell}$, define their \emph{$T$-inner product} by
    \[\langle\boldsymbol{a}, \boldsymbol{b}\rangle_T:= \boldsymbol{a} T T^T  \boldsymbol{b}^T =\langle \boldsymbol{a}T, \boldsymbol{b}T\rangle. \]

    This is a symmetric nondegenerate bilinear inner product. Furthermore, this inner product has good properties with respect to the dual code defined below.

    Similarly, for $\boldsymbol{\overline{a}}, \boldsymbol{\overline{b}}\in  P_n^{r\ell}$, define their \emph{$T$-inner product} by $\left\langle\boldsymbol{\overline{a}}, \boldsymbol{\overline{b}}\right\rangle_T:= \langle\boldsymbol{a}, \boldsymbol{b}\rangle_T+(X^n-1)\in P_n$, where $\pi: \F_{q^a}[X;\sigma]^{r\ell} \to P_n^{r\ell}$ denotes the canonical projection, and
    $\boldsymbol{\overline{a}}=\pi(\bs{a}), \boldsymbol{\overline{b}}=\pi(\bs{b})$ for some $\bs{a}, \bs{b}\in \F_{q^a}[X;\sigma]^{r\ell}$.
\end{definition}

\begin{definition} \label{star}
    Let $\C$ be an $\F_{q^a}$-linear skew quasi-cyclic $\F_{q^m}$-code of length $n\ell$ and index $\ell$ with decomposition as in Theorem \ref{thm:main_theorem},
    \begin{align}\label{codedecomposition}
        \C \simeq \left(\bigoplus_{i=1}^{r\ell} \F_{q^a} [X;\sigma]\boldsymbol{c}_i  \right) \Big/ \left(\bigoplus_{i=1}^{r\ell}  \F_{q^a} [X;\sigma]d_i\boldsymbol{c}_i  \right).
    \end{align}

    Define the \emph{$T$-dual of $\C$} (with respect to $\boldsymbol{q}_1, \ldots, \boldsymbol{q}_{r\ell}$) by
    \[\C_T^\perp := \left\langle \boldsymbol{c}_1',
        \ldots, \boldsymbol{c}_{r\ell}'  \right\rangle_{\F_{q^a}[X;\sigma]} +(X^n-1)\F_{q^a}[X;\sigma]^{r\ell}, \]
    where $\boldsymbol{c}_i = h_i \boldsymbol{q}_i$, $h_i d_i = d_i h_i = X^n-1$, and $\boldsymbol{c}_i' := d_i \boldsymbol{q}_i$.
    Note that the $\boldsymbol{c}_i'$ are linearly independent over $\F_{q^a}[X;\sigma]$ since the $\boldsymbol{q}_i$ form a basis for $\F_{q^a}[X;\sigma]^{r\ell}$.
\end{definition}

This dual has desirable properties. In particular, it allows us to recover the original code if we keep track of the basis elements $\boldsymbol{q}_i$.

\begin{theorem}\label{thm:gen_dual}
    Let $\C$ be an $\F_{q^a}$-linear skew quasi-cyclic $\F_{q^m}$-code of length $n\ell$ and index $\ell$, and suppose that $\C_{T}^\perp$ is defined with respect to the decomposition \eqref{codedecomposition} as above.
    \begin{enumerate}
        \item[$(1)$] If $\boldsymbol{\overline{x}}\in \C$ and $\boldsymbol{\overline{y}}\in  \C_T^\perp$, then $\left\langle\boldsymbol{\overline{x}}, \boldsymbol{\overline{y}}\right\rangle_T=0_{P_n}$.
        \item[$(2)$] When constructed with respect to the same basis $\{\boldsymbol{q}_1, \ldots, \boldsymbol{q}_{r\ell} \}$, we have $\left( \C_T^\perp \right)_T^\perp =\C.$
    \end{enumerate}
\end{theorem}

\begin{proof}
    ${(1)}$ We have that $\bs{q}_i T=\bs{e}_i$, where $\bs{e}_i$ is the row vector with entry 1 in the $i$-th position, and entries $0$ elsewhere.
    Let $\boldsymbol{\overline{c}}_i= \bs{c}_i+(X^n-1)\F_{q^a}[X;\sigma]^{r\ell}$ and $\boldsymbol{\overline{c}}_j'= \boldsymbol{c}_j'+(X^n-1)\F_{q^a}[X;\sigma]^{r\ell}$.
    Then for all $i,j\in \{1,\ldots, r\ell\}$,
    \begin{align*}
        \left\langle\boldsymbol{\overline{c}}_i,\boldsymbol{\overline{c}}_j'\right\rangle_T & = \langle\boldsymbol{c}_i, \boldsymbol{c}_j'\rangle_T + (X^n-1)      \\
                                                                                            & = \langle h_i\boldsymbol{q}_iT, d_j\boldsymbol{q}_jT\rangle +(X^n-1) \\
                                                                                            & = \langle h_i \bs{e}_i, d_j \bs{e}_j\rangle +(X^n-1)                 \\
                                                                                            & = h_i d_j \delta_{ij} +(X^n-1)                                       \\
                                                                                            & =(X^n-1)\delta_{ij} +(X^n-1) = 0_{P_n}.
    \end{align*}
    Since $\C =\langle \boldsymbol{\overline{c}}_1, \ldots, \boldsymbol{\overline{c}}_{r\ell}\rangle_{\F_{q^a}[X;\sigma]}$ and $\C_T^\perp =\langle \boldsymbol{\overline{c}}_1', \ldots, \boldsymbol{\overline{c}}_{r\ell}'\rangle_{\F_{q^a}[X;\sigma]}$, the result follows by bilinearity of the inner product.

    ${(2)}$ Let $i\in \{1, \ldots, r\ell\}$. Recall that $\boldsymbol{c}_i = h_i \boldsymbol{q}_i$ and $\boldsymbol{c}_i' = d_i \boldsymbol{q}_i$. To construct $\left( \C_T^\perp \right)_T^\perp$,
    Definition \ref{star} instructs us to decompose $\boldsymbol{c}_i' = d_i \boldsymbol{q}_i$ to
    define new basis vectors $\boldsymbol{c}_i'' := h_i \boldsymbol{q}_i= \boldsymbol{c}_i$ for $\left( \C_T^\perp \right)_T^\perp$.
    Hence $\left( \C_T^\perp \right)_T^\perp$ and $\C$ have the same basis.
\end{proof}

\section{Conclusion and Future Work}\label{sec:future}
This article introduces \emph{nonlinear skew quasi-cyclic codes} of length $n\ell$ and index $\ell$. These codes simultaneously generalize quasi-cyclic, skew cyclic, and skew quasi-cyclic codes. Further exploration of these codes shows that they correspond to the left $P_n$-submodules of the left $P_n$-module $R_n^\ell$. We then utilize a Smith normal form construction to give a classification of nonlinear skew quasi-cyclic codes, namely, Theorem \ref{thm:main_theorem}. This module-theoretic classification leads us to Theorem \ref{thm:gen_dual}, which provides a natural alternative for the standard Euclidean dual of a nonlinear skew quasi-cyclic code.

While the present article gives a definition for a dual code of a given $(\sigma,\ell)$-QC code, we do not know if this is the ``best'' definition of a dual code. The forthcoming article \cite{BHR_2} will continue an exploration of this topic.

\printbibliography

@inproceedings{prcyu,
  title = {On irreducible polynomial remainder codes},
  author = {Yu, J.-H. and Loeliger, H.-A.},
  booktitle = {2011 IEEE International Symposium on Information Theory Proceedings},
  pages = {1190--1194},
  year = {2011},
  organization = {IEEE}
}

@book{1997-lidl-nied,
    AUTHOR = {Lidl, R. and Niederreiter, H.},
     TITLE = {Finite Fields},
    SERIES = {Encyclopedia of Mathematics and its Applications~{ 20}},
   EDITION = {2nd},
 PUBLISHER = {Cambridge University Press, Cambridge},
      YEAR = {1997},
     PAGES = {xiv+755},
      ISBN = {0-521-39231-4},
   MRCLASS = {11Txx},
  MRNUMBER = {1429394},
}

@article{Huffman:2010,
    AUTHOR = {Huffman, W.C.},
     TITLE = {Cyclic {$\mathbb F_q$}-linear {$\mathbb F_{q^t}$}-codes},
   JOURNAL = {Int. J. Inf. Coding Theory},
    VOLUME = {1},
  FJOURNAL = {International Journal of Information and Coding Theory.
              IJICOT},
      YEAR = {2010},
     PAGES = {249--284},
   MRCLASS = {94B15 (81P70)},
  MRNUMBER = {2772898},
MRREVIEWER = {Anuradha\ Sharma},
       DOI = {10.1504/IJICOT.2010.032543},
}

@article {Abualrub-et-al:2010,
    AUTHOR = {Abualrub, T. and Ghrayeb, A. and Aydin, N. and Siap,
              I.},
     TITLE = {On the construction of skew quasi-cyclic codes},
   JOURNAL = {IEEE Trans. Inform. Theory},
  FJOURNAL = {Institute of Electrical and Electronics Engineers.
              Transactions on Information Theory},
    VOLUME = {56},
      YEAR = {2010},
     PAGES = {2081--2090},
      ISSN = {0018-9448,1557-9654},
   MRCLASS = {94B15 (94B60)},
  MRNUMBER = {2723473},
       DOI = {10.1109/TIT.2010.2044062},
}

@article {Abualrub-et-al:2018,
    AUTHOR = {Abualrub, T. and Ezerman, M. and
              Seneviratne, P. and Sol\'{e}, P.},
     TITLE = {Skew generalized quasi-cyclic codes},
   JOURNAL = {TWMS J. Pure Appl. Math.},
  FJOURNAL = {TWMS Journal of Pure and Applied Mathematics},
    VOLUME = {9},
      YEAR = {2018},
     PAGES = {123--134},
      ISSN = {2076-2585,2219-1259},
   MRCLASS = {94B05 (16S36 81P70)},
  MRNUMBER = {3823003},
}

@article {Bhaintwal:2012,
    AUTHOR = {Bhaintwal, M.},
     TITLE = {Skew quasi-cyclic codes over {G}alois rings},
   JOURNAL = {Des. Codes Cryptogr.},
  FJOURNAL = {Designs, Codes and Cryptography. An International Journal},
    VOLUME = {62},
      YEAR = {2012},
     PAGES = {85--101},
      ISSN = {0925-1022,1573-7586},
   MRCLASS = {94B15 (11T71 94B60)},
  MRNUMBER = {2873109},
       DOI = {10.1007/s10623-011-9494-0},
}

@article {Boucher-Geiselmann-Ulmer:2007,
    AUTHOR = {Boucher, D. and Geiselmann, W. and Ulmer, F.},
     TITLE = {Skew-cyclic codes},
   JOURNAL = {Appl. Algebra Engrg. Comm. Comput.},
  FJOURNAL = {Applicable Algebra in Engineering, Communication and
              Computing},
    VOLUME = {18},
      YEAR = {2007},
     PAGES = {379--389},
   MRCLASS = {94B25 (16P10)},
  MRNUMBER = {2322946},
       DOI = {10.1007/s00200-007-0043-z},
}

@article {Boucher-Ulmer:2009,
    AUTHOR = {Boucher, D. and Ulmer, F.},
     TITLE = {Coding with skew polynomial rings},
   JOURNAL = {J. Symbolic Comput.},
  FJOURNAL = {Journal of Symbolic Computation},
    VOLUME = {44},
      YEAR = {2009},
     PAGES = {1644--1656},
      ISSN = {0747-7171,1095-855X},
   MRCLASS = {94B40 (16S35)},
  MRNUMBER = {2553570},
MRREVIEWER = {Manish\ Kumar\ Gupta},
       DOI = {10.1016/j.jsc.2007.11.008}
}

@article {Gao-Shen-Fu:2016,
    AUTHOR = {Gao, J. and Shen, L. and Fu, F.},
     TITLE = {A {C}hinese remainder theorem approach to skew generalized
              quasi-cyclic codes over finite fields},
   JOURNAL = {Cryptogr. Commun.},
  FJOURNAL = {Cryptography and Communications. Discrete Structures, Boolean
              Functions and Sequences},
    VOLUME = {8},
      YEAR = {2016},
     PAGES = {51--66},
      ISSN = {1936-2447,1936-2455},
   MRCLASS = {94B05 (94B15)},
  MRNUMBER = {3480616},
       DOI = {10.1007/s12095-015-0140-y},
}

@article {Lally-Fitzpatrick:2001,
    AUTHOR = {Lally, K. and Fitzpatrick, P.},
     TITLE = {Algebraic structure of quasicyclic codes},
   JOURNAL = {Discrete Appl. Math.},
  FJOURNAL = {Discrete Applied Mathematics. The Journal of Combinatorial
              Algorithms, Informatics and Computational Sciences},
    VOLUME = {111},
      YEAR = {2001},
     PAGES = {157--175},
      ISSN = {0166-218X,1872-6771},
   MRCLASS = {94B05 (13P10)},
  MRNUMBER = {1836725},
MRREVIEWER = {Edgar\ Mart\'{\i}nez-Moro},
       DOI = {10.1016/S0166-218X(00)00350-4},
}

@article {Ling-Sole:2001,
    AUTHOR = {Ling, S. and Sol\'{e}, P.},
     TITLE = {On the algebraic structure of quasi-cyclic codes. {I}.
              {F}inite fields},
   JOURNAL = {IEEE Trans. Inform. Theory},
  FJOURNAL = {Institute of Electrical and Electronics Engineers.
              Transactions on Information Theory},
    VOLUME = {47},
      YEAR = {2001},
     PAGES = {2751--2760},
      ISSN = {0018-9448,1557-9654},
   MRCLASS = {94B15 (94B40)},
  MRNUMBER = {1872837},
       DOI = {10.1109/18.959257},
}

@article {Ling-Sole:2003,
    AUTHOR = {Ling, S. and Sol\'{e}, P.},
     TITLE = {On the algebraic structure of quasi-cyclic codes. {II}.
              {C}hain rings},
   JOURNAL = {Des. Codes Cryptogr.},
  FJOURNAL = {Designs, Codes and Cryptography. An International Journal},
    VOLUME = {30},
      YEAR = {2003},
     PAGES = {113--130},
      ISSN = {0925-1022,1573-7586},
   MRCLASS = {94B05 (11H71 94B15)},
  MRNUMBER = {1998855},
       DOI = {10.1023/A:1024715527805},
}

@article {Ling-Sole:2005,
    AUTHOR = {Ling, S. and Sol\'{e}, P.},
     TITLE = {On the algebraic structure of quasi-cyclic codes. {III}.
              {G}enerator theory},
   JOURNAL = {IEEE Trans. Inform. Theory},
  FJOURNAL = {Institute of Electrical and Electronics Engineers.
              Transactions on Information Theory},
    VOLUME = {51},
      YEAR = {2005},
     PAGES = {2692--2700},
      ISSN = {0018-9448,1557-9654},
   MRCLASS = {94B15 (11H71 94B05)},
  MRNUMBER = {2246389},
MRREVIEWER = {Yuan\ Luo},
       DOI = {10.1109/TIT.2005.850142},
}

@article {Ling-Sole:2006,
    AUTHOR = {Ling, S. and Niederreiter, H. and Sol\'{e}, P.},
     TITLE = {On the algebraic structure of quasi-cyclic codes. {IV}.
              {R}epeated roots},
   JOURNAL = {Des. Codes Cryptogr.},
  FJOURNAL = {Designs, Codes and Cryptography. An International Journal},
    VOLUME = {38},
      YEAR = {2006},
     PAGES = {337--361},
      ISSN = {0925-1022,1573-7586},
   MRCLASS = {94B05 (11T71 94B15)},
  MRNUMBER = {2195520},
MRREVIEWER = {Edgar\ Mart\'{\i}nez-Moro},
       DOI = {10.1007/s10623-005-1431-7},
}

@article {Ou-azzou-et-al:2024,
    AUTHOR = {Ou-azzou, H. and Najmeddine, M. and Aydin, N. and
              Mouloua, E.},
     TITLE = {On the algebraic structure of {$(M,\sigma,\delta)$}-skew
              codes},
   JOURNAL = {J. Algebra},
  FJOURNAL = {Journal of Algebra},
    VOLUME = {637},
      YEAR = {2024},
     PAGES = {156--192},
      ISSN = {0021-8693,1090-266X},
   MRCLASS = {16S36 (94B15)},
  MRNUMBER = {4651847},
       DOI = {10.1016/j.jalgebra.2023.09.014},
}

@incollection {Gluesing-Luerssen:2021,
    AUTHOR = {Gluesing-Luerssen, H.},
     TITLE = {Introduction to skew-polynomial rings and skew-cyclic codes},
 BOOKTITLE = {Concise Encyclopedia of Coding Theory},
     PAGES = {151--180},
 PUBLISHER = {CRC Press, Boca Raton, FL},
      YEAR = {2021},
   MRCLASS = {94B15},
  MRNUMBER = {4598740},
}

@article {Huffman:2013,
    AUTHOR = {Huffman, W.C.},
     TITLE = {On the theory of {$\mathbb F_q$}-linear {$\mathbb F_{q^t}$}-codes},
   JOURNAL = {Adv. Math. Commun.},
  FJOURNAL = {Advances in Mathematics of Communications},
    VOLUME = {7},
      YEAR = {2013},
     PAGES = {349--378},
   MRCLASS = {94B05 (81P70 94B27 94B60)},
  MRNUMBER = {3092314},
MRREVIEWER = {Rumen\ N.\ Daskalov},
       DOI = {10.3934/amc.2013.7.349},
}

@article {Sharma-Kaur:2017,
    AUTHOR = {Sharma, A. and Kaur, T.},
     TITLE = {On cyclic {$\mathbb F_q$}-linear {$\mathbb F_{q^t}$}-codes},
   JOURNAL = {Int. J. Inf. Coding Theory},
  FJOURNAL = {International Journal of Information and Coding Theory.
              IJICOT},
    VOLUME = {4},
      YEAR = {2017},
     PAGES = {19--46},
      ISSN = {1753-7703,1753-7711},
   MRCLASS = {94B15},
  MRNUMBER = {3602545},
       DOI = {10.1504/IJICOT.2017.081457},
}

@article{BLESSENOHL-KARSTENJOHNSEN:1986,
title = {Eine {V}erschärfung des {S}atzes von der {N}ormalbasis},
journal = {J. Algebra},
volume = {103},
pages = {141-159},
year = {1986},
issn = {0021-8693},
doi = {10.1016/0021-8693(86)90174-2},
author = {D. Blessenohl and K. Johnsen}
}

@article{rs1960,
  title={Polynomial Codes Over Certain Finite Fields},
  author={Reed, I.S. and Solomon, G.},
  journal={J. SIAM},
  volume={8},
  number={2},
  pages={300--304},
  year={1960},
  publisher={Society for Industrial and Applied Mathematics}
}

@book{Berrick:2000,
  title={An Introduction to Rings and Modules},
  author={Berrick, A.J. and Keating, M.E.},
  year={2000},
  publisher={Cambridge University Press},
  address={Cambridge}
}

@book{conciseencyclopedia,
  title={Concise Encyclopedia of Coding Theory},
  author={Huffman, W. C. and Kim, J.-L. and Solé, P.},
  year={2021},
  publisher={CRC Press}
}

@unpublished{BHR_2,
    AUTHOR = {Bossaller, D. and Herden, D. and Ruiz-Bola\~nos, I.},
    TITLE = {The Trace Dual of Nonlinear Skew Cyclic Codes},
    year = {2025},
    note = {Submitted}
}

@book{Macwilliams-Sloane:1977,
  title={The {T}heory of {E}rror-{C}orrecting {C}odes},
  author={MacWilliams, F.J. and Sloane, N.J.A.},
  series={North-Holland Mathematical Library},
  volume={16},
  year={1977},
  publisher={Elsevier}
}

@article{Ore:1933,
  title={Theory of non-commutative polynomials},
  author={Ore, O.},
  journal={Annals of Mathematics},
  volume={34},
  number={3},
  pages={480--508},
  year={1933},
  publisher={JSTOR}
}

@article{Lam-Leroy:1988,
  title={Vandermonde and {W}ronskian matrices over division rings},
  author={Lam, T.-Y. and Leroy, A.},
  journal={J. Algebra},
  volume={119},
  number={2},
  pages={308--336},
  year={1988},
  publisher={Academic Press}
}

@article{Yadav-2024,
title = {Some new classes of additive {MDS} and almost {MDS} codes over finite fields},
journal = {Finite Fields Appl.},
volume = {95},
pages = {102394},
year = {2024},
issn = {1071-5797},
doi = {10.1016/j.ffa.2024.102394},
author = {Yadav, M. and Sharma, A.},
}

@article{Bhunia2025,
  title={Foundations of additive codes over finite fields},
  author={Bhunia, D.K. and Dougherty, S.T. and Fern{\'a}ndez-C{\'o}rdoba, C. and Villanueva, M.},
  journal={Finite Fields Appl.},
  volume={104},
  pages={Paper No. 102592},
  year={2025},
  publisher={Elsevier}
}

@article{delsarte1971,
  title={Majority logic decodable codes derived from finite inversive planes},
  author={Delsarte, P.},
  journal={Inform. and Control},
  volume={18},
  pages={319--325},
  year={1971},
  publisher={Academic Press}
}

@article{Ball:2023,
  title={On additive {MDS} codes over small fields},
  author={Ball, S. and Gamboa, G. and  Lavrauw, M.},
  journal={Adv. Math. Commun.},
  volume={17},
  pages={828--844},
  year={2023}
}

@article{Ball:2025,
  title={Griesmer type bounds for additive codes over finite fields, integral and fractional {MDS} codes},
  author={Ball, S. and Lavrauw, M. and Popatia, T.},
  journal={Designs, Codes and Cryptography},
  volume={93},
  pages={175--196},
  year={2025}
}

@misc{Martinez:2024,
Author = {Umberto Martínez-Peñas and Rubén Rodríguez-Ballesteros},
Title = {Linear codes in the folded {H}amming distance and the quasi {MDS} property},
Year = {2024},
Eprint = {2406.13355},
}

@book{Jacobson:2009,
  title={Finite-{D}imensional {D}ivision {A}lgebras over {F}ields},
  author={Jacobson, N.},
  year={1996},
  publisher={Springer}
}
\end{document}